\documentclass[a4paper,english]{svjour3}
\usepackage[LGR,T2A,LGR,T2A,T1]{fontenc}
\usepackage[latin9]{inputenc}
\usepackage{color}
\usepackage{textcomp}
\usepackage{amsmath}
\usepackage{amssymb}
\usepackage{esint}

\makeatletter

\DeclareRobustCommand{\greektext}{%
  \fontencoding{LGR}\selectfont\def\encodingdefault{LGR}}
\DeclareRobustCommand{\textgreek}[1]{\leavevmode{\greektext #1}}
\DeclareFontEncoding{LGR}{}{}
\DeclareTextSymbol{\~}{LGR}{126}
\DeclareRobustCommand{\cyrtext}{%
  \fontencoding{T2A}\selectfont\def\encodingdefault{T2A}}
\DeclareRobustCommand{\textcyr}[1]{\leavevmode{\cyrtext #1}}
\AtBeginDocument{\DeclareFontEncoding{T2A}{}{}}

\newcommand{\lyxmathsym}[1]{\ifmmode\begingroup\def\b@ld{bold}
  \text{\ifx\math@version\b@ld\bfseries\fi#1}\endgroup\else#1\fi}

\usepackage{babel}

\makeatother

\usepackage{babel}

\begin{document}

\title{Earthquake forecasting: Statistics and Information}

\author{V.Gertsik {\normalsize{{}}}%
\thanks{ Institute of Earthquake Prediction Theory and Mathematical Geophysics
RAS, Moscow, RF, getrzik@ya.ru%
}, M.Kelbert{\normalsize{{}}}%
\thanks{Dept. of Mathematics, Swansea University, Singleton Park, Swansea,
SA2 8PP, UK. Institute of Earthquake Prediction Theory and Mathematical
Geophysics RAS, m.kelbert@swansea.ac.uk %
}, A.Krichevets{\normalsize{{}}}%
\thanks{Lomonosov MSU, Department of Psychology, Moscow, Russia, ankrich@mail.ru%
}}
\maketitle
\begin{abstract}
We present an axiomatic approach to earthquake forecasting in terms
of multi-component random fields on a lattice. This approach provides
a method for constructing point estimates and confidence intervals
for conditional probabilities of strong earthquakes under conditions
on the levels of precursors. Also, it provides an approach for setting
multilevel alarm system and hypothesis testing for binary alarms.
We use a method of comparison for different earthquake forecasts in
terms of the increase of Shannon information. 'Forecasting' and 'prediction'
of earthquakes are equivalent in this approach. 
\end{abstract}

\section{\textcolor{black}{Introduction}}

The methodology of selecting and processing of relevant information
about the future occurrence of potentially damaging earthquakes has
reached a reasonable level of maturity over the recent years. However,
the problem as a whole still lacks a comprehensive and generally accepted
solution. Further efforts for optimization of the methodology of forecasting
would be productive and well-justified.

A comprehensive review of the modern earthquake forecasting state
of knowledge and guidelines for utilization can be found in {[}Jordan
et all., 2011{]}. Note that all methods of evaluating the probabilities
of earthquakes are based on a combination of geophysical, geological
and probabilistic models and considerations. Even the best and very
detailed models used in practice are in fact only 'caricatures' of
immensely complicated real processes.

A mathematical toolkit for earthquake forecasting is well presented
in the paper {[}Harte and Vere-Jones, 2005{]}. This work is based
on the modeling of earthquake sequences in terms of the marked point
processes. However, the mathematical technique used is quite sophisticated
and does not provide direct practical tools to investigate the relations
of the structure of temporal-spatial random fields of precursors to
the appearance of strong earthquakes.

The use of the multicomponent lattice models (instead of marked point
processes) gives a different/novel way of investigating these relations
in a more elementary way. Discretization of space and time allows
us to separate the problem in question into two separate tasks. The
first task is the selection of relevant precursors, i.e., observable
and theoretically explained physical and geological facts which are
casually related to a high probability of strong earthquakes. Particularly,
this task involves the development of models of seismic events and
computing probabilities of strong earthquakes in the framework of
these models. Such probabilities are used as precursors in the second
task.

The second task is the development of methodology of working with
these precursors in order to extract the maximum information about
the probabilities of strong earthquakes. This is the main topic of
this paper.

Our approach allows us to obtain the following results:

$\bullet$ Estimates of probabilities of strong earthquakes for given
values of precursors are calculated in terms of the frequencies of
historic data.

$\bullet$ Confidence intervals are also constructed to provide reasonable
bounds of precision for point estimates.

$\bullet$ Methods of predictions (i.e., binary alarm announcement
{[}Keilis-Borok, 1996{]}, {[}Keilis-Borok, Kossobokov, 1990{]}, {[}Holiday
et all., 2007{]}) and forecasting (i.e., calculating probabilities
of earthquakes {[}Jordan et all., 2011{]}, {[}Kagan and Jackson, 2000{]},
{[}Harte and Vere-Jones, 2005{]}, {[}WGNCEP{]}) are equivalent in
the following sense: the setting of some threshold for probability
of earthquakes allows to update the alarm level. On the other hand,
the knowledge of the alarm domain based on historical data allows
us to evaluate the probabilities of earthquakes. In a sense, the prediction
is equivalent to hypothesis testing as well, see Section 11.

$\bullet$ In our scheme we propose a scalar statistic which is the
ratio of actual increment of information to the maximal possible increment
of information. This statistic allows us to linearly order all possible
forecasting algorithms. Nowadays the final judgement about the quality
of earthquake forecasting algorithms is left to experts. This arrangement
puts the problem outside the scope of natural sciences which are trying
to avoid subjective judgements.

The foundation of our proposed scheme is the assumption that the seismic
process is random and cannot be described by a purely deterministic
model. Indeed, if the seismic process is deterministic then the inaccuracy
of the forecast could be explained by the non-completeness of our
knowledge about the seismic events and non-precision of the available
information. This may explain, at least in principle, attacks from
the authorities addressed to geophysicists who failed to predict a
damaging earthquake. However, these attacks have no grounds if one
accepts that the seismic process is random. At the end of the last
century (February-April 1999) a group of leading seismologists organized
a debate via the web to form a collective opinion of the scientific
community on the topic: 'Is the reliable prediction of individual
earthquakes a realistic scientic goal?' (see http://www.nature.com/nature/debates/earthquake/).

Despite a considerable divergence in peripheral issues all experts
taking part in the debate agreed on the following main principles:

$\bullet$ the deterministic prediction of an individual earthquake,
within sufficiently narrow limits to allow a planned evacuation programme,
is an unrealistic goal;

$\bullet$ forecasting of at least some forms of time-dependent seismic
hazard can be justified on both physical and observational grounds.

The following facts form the basis of our agreement with this point
of view.

The string-block Burridge-Knopov model, generally accepted as a mathematical
tool to demonstrate the power-like Gutenberg-Richter relationship
between the magnitude and the number of earthquakes, involves the
generators of chaotic behaviour or dynamic stochasticity. In fact,
the nonlinearity makes the seismic processes stochastic: a small change
in the shift force may lead to completely different consequences.
If the force is below the threshold of static friction the block is
immovable, if the force exceeds this threshold it starts moving, producing
an avalanche of unpredictable size.

This mechanism is widespread in the Earth. Suppose that the front
propagation of the earthquake approaches a region of enhanced strength
of the rocks. The earthquake magnitude depends on whether this region
will be destroyed or remains intact. In the first case the front moves
further on, in the second case the earthquake remains localized. So,
if the strength of the rocks is below the threshold the first scenario
prevails, if it is above the threshold the second scenario is adapted.
The whole situation is usually labelled as a \textit{butterfly effect}:
infinitesimally small changes of strength and stress lead to macroscopic
consequences which cannot be predicted because this infinitesimal
change is below any precision of the measurement. For these reasons
determinism of seismic processes looks more doubtful than stochasticity.

The only comment we would like to contribute to this discussion is
that the forecasting algorithms based exclusively on the empirical
data without consistent physical models could hardly be effective
in practice (see Sections 12, 13 for more details).

In conclusion we discuss the problem of precursor selection and present
a theorem by A. Krichevetz stating that using a learning sample for
an arbitrary feature selection in pattern recognition is useless in
principle.

Finally, note that our approach may be well-applicable for the space-time
forecasting of different extremal events outside the scope of earthquake
prediction.

\section{Events and precursors on the lattice}

In order to define explicitly estimates of probabilities of strong
earthquakes we discretize the two-dimensional physical space and time,
i.e., introduce a partition of three-dimensional space-time into rectangular
cells with the space partition in the shape of squares and time partition
in the shape of intervals. These cells should not intersect to avoid
an ambiguity in computing the frequencies for each cell. In fact,
the space cells should not be perfect squares because of the curvature
of the earth's surface, but this may be neglected if the region of
forecasting is not too large.

So, we obtain a discrete set $\Omega_{K}$ with $N=I\times J\times K$
points which is defined as follows. Let us select a rectangular domain
$A$ of the two-dimensional lattice with $I\times J$ points $x=(x_{i},y_{j});\, x_{i}=a\text{\texttimes}i;\, i=1,\ldots,I$
and $y_{j}=a\text{\texttimes}j;\, j=1,\ldots,J,$$a$ is the step
of the lattice. A cell in $\Omega_{K}$ takes the shape of parallelepiped
of height $\lyxmathsym{\textgreek{D}}t$ with a square base. Clearly,
any point in $\Omega_{K}$ has coordinates $(x_{i},y_{j},t_{k}),\, t_{k}=t+k\Delta t;\, k=0,\ldots,K.$
.

We say that \emph{a seismic event} happens if an earthquake with magnitude
greater than some pre-selected threshold $M_{0}$ is registered, and
this earthquake is not foreshock or aftershock of another, more powerful
earthquake (we put aside a technical problem of identification of
foreshocks and aftershocks in the sequence of a seismic event). For
any cell in our space-time grid we define \textit{an indicator of
an event}, i.e., a binary function $h(i,j,k)$. This function takes
the value $1$ if at least one seismic event is registered in a given
cell and $0$ otherwise. Suppose that for all points $(x_{i},y_{j},t_{k})$
the value of a vector precursor $\mathbf{\boldsymbol{f}}(i,j,k)=\left\{ f_{q}(i,j,k),\, q=1,...,Q\right\} $
is given. The components of the precursor $f_{q}(i,j,k),\, q=1,...,Q$
are the scalar statistics constructed on the base of our understanding
of the phenomena that precede a seismic event.
\begin{remark}
Note that specifying an alarm domain as a circle with center at a
lattice site and radius proportional to the maximal magnitude of the
forecasted earthquakes leads to a contradiction. Indeed, suppose we
announce an alarm for earthquakes with minimal magnitude $6$ in a
domain $A_{6}$. Obviously, the same alarm should be announced in
the domain $A_{7}$ as well. By the definition $A_{6}\subset A_{7}$
and we expect an earthquake with magnitude at least $7$ and do not
expect an earthquake with magnitude at least $6$ in the domain $A_{7}\setminus A_{6}$.
But this is absurd.
\end{remark}

\section{Mathematical assumptions}

A number of basic assumptions form the foundation of the mathematical
tecnique of earthquake forecasting. In the framework of mathematical
theory they can be treated as axioms but are, in fact, an idealization
and simplification with respect to the description of the real phenomena.
Below we summarize the basic assumptions which are routinely used
in existing studies of seismicity and algorithms of earthquake forecasting
even the authors do not always formulate them explicitly.

We accept the following assumptions or axioms of the mathematical
theory:

\noindent (i) The multicomponent random process $\left\{ h(i,j,k),\mathbf{f}(i,j,k)\right\} $,
describing the joint evolution of the vector precursors and the indicator
of seismic events, is stationary.

This assumption provides an opportunity to investigate the intrinsic
relations between the precursors and the seismic events based on the
historical data. In other words, the experience obtained by analysing
the series of events in the past, is applicable to the future as the
properties of the process do not depend on time.

In reality, this assumption holds only approximately and for a restricted
time period. Indeed, plate tectonics destroys the stationarity for
a number of reasons including some purely geometrical considerations.
For instance, the movements of the plates leads to their collisions,
their partial destruction and also changes their shapes. Nevertheless,
the seismic process can be treated as quasi-stationary one for considerable
periods of time. At the time when the system changes one quasi-stationary
regime to another (say, nowadays, many researcher speak about the
abrupt climate change) the reliability of any prediction including
the forecast of seismic events is severely restricted.

\noindent (ii) The multicomponent random process $\left\{ h(i,j,k),\mathbf{f}(i,j,k)\right\} $
is ergodic.

Any quantitative characteristic of seismicity more representative
than a registration of an individual event is, in fact, the result
of averaging over time. For instance, the Gutenberg-Richter law, applied
to a given region relates the magnitude with the average number of
earthquakes where the averaging is taken over a specific time interval.
In order to associate with the time averaging a proper probabilistic
characteristic of the process and make a forecast about the future
one naturally needs the assumption of ergodicity. This exactly means
that any averaging over time interval $[0,T]$ will converge to the
stochastic average when $T\to\infty$. In view of ergodicity one can
also construct unbiased and consistent estimates of conditional probabilities
of strong earthquakes under conditions that the precursors take values
in some intervals. Naturally, these estimates are the frequencies
of observed earthquakes, i.e., ratios of the number of cells with
seismic events and prescribed values of precursors to the total number
of cells with the prescribed values of precursors. (Recall that an
unbiased point estimate ${\hat{\theta}}$ of parameter $\theta$ satisfies
the condition ${\bf E}{\hat{\theta}}=\theta$, and a consistent estimate
converges to the true value $\theta$ when the sample size tends to
infinity).

\noindent (iii) Any statement about the value of the indicator of
a seismic event $h(i,j,k)$ in the cell $(i,j,k)$ or its probability
should be based on the values of the precursor ${\bf f}(i,j,k)$ only.

This assumption means that the precursor in the given cell accumulates
all the relevant information about the past and the information about
the local properties of the area that may be used for the forecast
of the seismic event in this cell. In other words, the best possible
precursor is used (which is not always the case in practice). As in
the other cases, this assumption is only an approximation to reality,
and the quality of a forecast depends on the quality of the selection
and accumulation of relevant information in the precursors.

Below we present some corollaries and further specifications.

\noindent (iii-a) For any $k$ the random variables $h(i,j,k),\, i=1,...,I,\, j=1,\ldots,J$
are conditionally independent under the condition that the values
of any measurable function $u(\mathbf{\boldsymbol{f}}(i,j,k))$ of
the precursors ${\bf f}(i,j,k),\, i=1,\ldots,I,\, j=1,\ldots,J$ are
fixed..

In practice this assumption means that the forecast for the time $t_{k}=t_{0}+k\Delta t$
cannot be affected by the values related to the future time intervals
$(t_{k},t_{k}+\Delta t]$. In reality all of these events may be dependent,
but our forecast does not use the information from the future after
$t_{k}$.

\noindent (iii-b) The conditional distribution of the random variable
\emph{$h$ } at a given cell depends on the values of the precursors
at this cell \emph{$\boldsymbol{\mathbf{f}}$} and is independent
of all other variables.

\noindent (iii-c) The conditional probabilities $\mathrm{Pr}_{ij}\left\{ \mathit{h\mid u\mathrm{(}\boldsymbol{\mathbf{f}}\mathrm{)}}\right\} $
of the indicator of seismic events $h$ in the cell $(i,j,k)$, under
condition $u\mathrm{(}\boldsymbol{\mathbf{f}}\mathrm{)}$ in this
cell do not depend on the position of the cell in space (the time
index $k$ related to this probablity may be dropped due to the stationarity
of the process).

In other words, the rule for computing the conditional probability
$\mathrm{Pr}_{ij}\left\{ \mathit{h\mid u\mathrm{(}\boldsymbol{\mathbf{f}}\mathrm{)}}\right\} $
based on the values of precursors is the same for all cells, and the
space indices of probability $\mathrm{Pr}$ may be dropped. This condition
is widely accepted in constructions of the forecasting algoritms but
rarely formulated explicitly. However, the probability of a seismic
event depends to a large extent on the local properties of the area.
Hence, the quality of the forecasting depends on how adequately these
properties are summarized in the precursors. This formalism properly
describes the space inhomogenuity of the physical space because the
stationary joint distribution of $\mathrm{Pr}_{ij}\left\{ \mathit{h\text{,}f\mathrm{\leq x}}\right\} $
for an arbitrary precursor $f$ depends, in general, on the position
of the cell in the domain $A.$ Below we will use the distributions
of precursors and indicators of seismic events in domain $A$ that
do not depend on the spatial coordinates and have the following form

\[
\mathrm{Pr}_{A}\left\{ \mathit{h\text{,}f\mathrm{\leq x}}\right\} =\frac{1}{\mathit{I\cdot J}}\sum_{\left(\mathit{i,j}\right)\in A}\mathbf{\mathrm{Pr}_{\mathrm{\mathit{ij}}}}\left\{ \mathit{h\text{,}f\mathrm{\leq x}}\right\} ,
\]

\[
\mathrm{\mathit{P}_{\mathit{A}}(\mathit{x})\equiv}\mathbf{\mathrm{Pr}_{\mathit{A}}}\left\{ \mathit{f\mathrm{\leq\mathit{x}}}\right\} =\frac{1}{\mathit{I\cdot J}}\sum_{\left(\mathit{i,j}\right)\in A}\mathbf{\mathrm{Pr}_{\mathrm{\mathit{ij}}}}\left\{ \mathit{f\leq\mathit{x}}\right\} ,
\]

\[
\mathrm{\mathit{p_{A}\equiv}}\mathrm{Pr}_{A}\left\{ \mathit{h\mathrm{=1}}\right\} =\frac{1}{\mathit{I\cdot J}}\sum_{\left(\mathit{i,j}\right)\in A}\mathbf{\mathrm{Pr}_{\mathrm{\mathit{ij}}}}\left\{ \mathit{h=\mathrm{1}}\right\} ,
\]

\noindent (iii-d) Note that assumption (iii) implies that the conditional
probabilities ${\bf Pr}\left(h\vert u({\bf f})\right)$ are computed
via the probabilities $\mathit{\mathbf{\mathrm{P\text{r}_{\mathit{A}}}}}\left\{ \mathit{h\text{,}f\mathrm{\leq x}}\right\} $
only.

The properties listed above are sufficient to obtain the point estimates
for the conditional probabilities of seismic events under conditions
formulated in terms of the values of precursors. However, additional
assumption are required for a testing of the forecasting algorithm:

\noindent (iv) The random variables $\boldsymbol{\mathbf{f}}\mathrm{(}i,j,k),$
are conditionally independent under condition that $h(i,j,k)=1$.

Again, these conditions are not exactly true, however they may be
treated as a reasonable approximation to reality. Indeed, if the threshold
$M_{0}$ is sufficiently high than the strong earthquakes may be treated
as rare events, and the cells where they are observed are far apart
with a high probability. Any two events related to cells separated
by the time intervals $\Delta t$ are asymptotically independent as
$\Delta t\rightarrow\infty$ because the seismic process has decaying
correlations (the mixing property in the language of random processes).
The loss of dependence (or decaying memory) is related to the physical
phemonema such as healing of the defects in the rocks, relaxation
of strength due to viscosity, etc. As usual in physical theories,
we accept an idealized model of the real phenomena applying this asymptotic
property for large but finite intervals between localizations of seismic
events.

The independence of strong earthquakes is not a new assumption, in
the case of continuous space-time it is equivalent to the assumption
that the locations of these events form a Poisson random field.. (Note
that the distribution of strong earthquake should be homogeneous in
space, because there is no information about the heterogeneity a priori
.) The Poisson hypothesis is used in many papers, see, e.g. {[}Harte
and Vere-Jones , 2005{]}. It is very natural for the analysis of the
«tails» of the Gutenberg-Richter law for large magnitudes {[}Pisarenko
et al., 2008{]}. Summing up, the development of the strict mathematical
theory of earthquake forecasting does not require any additional assumption
except those routinely accepted in the existing algorithms but usually
not formulated explicitly.

\section{The standard form for precursors}

The correct solution of the forecasting problem given the values of
precursors $\mathbf{f}(i,j,k)=\big(f_{1}(i,j,k),\ldots,f_{Q}(i,j,k)\big)$
is provided by the estimate of conditional probability $\text{Pr\ensuremath{\left\{ \mathit{h\mid\mathbf{f}(i,j,k)}\right\} }}$
of the indicator of seismic event in the cell $(i,j,k)$. In practice
this solution may be difficult to obtain because the number of events
in catalog is not sufficient.

Indeed, the range of value of a scalar precursor is usually divided
into a number $M$ of intervals, and only a few events are registered
for any such interval. For a $Q$-dimensional precursor the number
of $Q$-dimensional rectangles, covering the range, is already $M^{Q}$,
and majority of them contains $0$ event. Only a small number of such
rectangles contains one or more events, that is the precision of such
an estimate of conditional probability is usually too low to have
any practical value.

For this reason one constructs a new scalar precursor in the form
of the scalar function of component of the vector precursor, and optimize
its predictive power. This approach leads to additional complication
as the units of measurement and the physical sense of different components
of precursor are substantially different. In order to overcome this
problem one uses some transformation to reduce all the components
of the precursor to a standard form with the same sense and range
of values.

Let us transform all the precursors $f_{q}(i,j,k),\, q=1,...,Q$ to
variables with the values in $[0,1]$ providing estimates of conditional
probabilities. So, after some transformation $F$ we obtain an estimate
of $\mathrm{Pr}\left\{ h=\mathrm{1}\mid u({\bf f}(i,j,k))=1\right\} $,
where $u$ is a characteristic function of some interval ${\bf B}$,
i.e., the probability of event $h(i,j,k)=1$ under condition that
this precursor takes the value ${\bf f}(i,j,k)\in{\bf B}$.

The transformation $F$ of a scalar precursor $f(i,j,k)$ is defined
as follows. Fix an arbitrary small number $\varepsilon$. Let $L$
be a number of cells $(i,j,k)$ such that $h(i,j,k)=1$, and $Z_{l},\, l=1,\ldots,L,$
be the ordered statistics, i.e., the values $f(i,j,k)$ in these cells
listed in non-decreasing order. Define a new sequence $z_{m},_{\,}m=0,...,M,$
from the ordered statistics $Z_{l}$ by the following recursion: $z_{0}=-\infty$,
$z_{m}$ is defined as the first point in the sequence $Z_{l}$, such
that $z_{m}-z_{m-1}\geq\varepsilon$. Next, construct the sequence
$z_{m}^{*}=z_{m}+(z_{m+1}-z_{m})/2,$$\, m=1,\ldots,M-1$, and add
the auxiliary elements $z_{0}^{*}=-\infty,\: z_{M}^{*}=\infty$. Define
also a sequence $n_{m},\, m=1,\ldots,M$, where $n_{m}$ equals to
the number of values in the sequence $Z_{l}$, such that $z_{m-1}^{*}\leq Z_{l}<z_{m}^{*}$.
Finally, define the numbers $N_{m},\, m=1,...,M$ counting all cells
such that $z_{m-1}^{*}\leq f(i,j,k)<z_{m}^{*},\, m=1,...M$. Observe
that $\sum_{m=1}^{M}n_{m}=L,\:\sum_{m=1}^{M}N_{m}=N$, and use the
ratios

\begin{equation}
\lambda=\frac{L}{N}\label{eq:7}
\end{equation}
as estimate of unconditional probability of a seismic event in a given
cell

\begin{equation}
p_{A}\equiv\mathrm{Pr_{\mathit{A}}}\left\{ \mathit{h\mathrm{(}i,j,k\mathrm{)}=\mathrm{1}}\right\} =\int_{-\infty}^{\infty}\mathrm{Pr}\left\{ \mathit{h=\mathrm{1\mid\mathit{x}}}\right\} dP_{A}(x).\label{eq:8}
\end{equation}

The transformation $F$ is defined as follows

\begin{gather}
g=Ff(i,j,k)=\frac{n_{m}}{N_{m}}\text{,}\:\text{if}\: z_{m-1}^{*}\leq f(i,j,k)<z_{m}^{*},\, m=1,\ldots,M.\label{eq:1}
\end{gather}

This definition implies that transformation $F$ replace the value
of precursor for the frequency, i.e., the ratio of a number of cells
containing a seismic event and the values of precursor from $[z_{m-1}^{*},z_{m}^{*})$
to the number of cells with the value of precursor in this range.
These frequencies are the natural estimates of conditional probabilities
$\mathrm{Pr}\left\{ \mathit{h=\mathrm{1\mid}z_{m-1}^{*}\leq f<z_{m}^{*}}\right\} ,\, m=1,\ldots,M$,
computed with respect to stationary distribution $P_{A}\mathrm{(}x\mathrm{)}$:

\begin{equation}
\mathrm{Pr}\left\{ \mathit{h=\mathrm{1\mid}z_{m-1}^{*}\leq f<z_{m}^{*}}\right\} =\frac{\int_{z_{m-1}^{*}}^{z_{m}^{*}}\mathrm{Pr}\left\{ \mathit{h=\mathrm{1\mid\mathit{x}}}\right\} dP_{\Omega}(x)}{\int_{z_{m-1}^{*}}^{z_{m}^{*}}dP_{A}(x)}.\label{eq:2}
\end{equation}
(This conditional probability can be written as $\mathrm{Pr}\left\{ \mathit{h=\mathrm{1\mid}u\mathrm{(}f)}\right\} ,$
where $u$ is the characteristic function of interval $[z_{m-1}^{*},z_{m}^{*})$).
The function $g$ has a stepwise shape, and the length of the step
in bounded from below by $\varepsilon$. It can be checked that there
exist the limit $\tilde{g}=\underset{\varepsilon\rightarrow0}{\lim}\underset{K\rightarrow\infty}{\lim}g=\mathrm{Pr}\left\{ \mathit{h=\mathrm{1\mid}f}\right\} .$

The estimates of conditional probabilities in terms of the function
$g$ are quite rough because typically the numbers $n_{m}$ are of
the order 1. As a final result we will present below more sharp but
less detailed estimates of conditional probabilities and confidence
intervals for them.

\section{Combinations of precursors}

There are many ways to construct a single scalar precursor based on
the vector precursor $(Ff_{q},\, q=1,\ldots,Q)$. Each such construction
inevitably contains a number of parameters or degrees of freedom.
These parameters (including the parameters used for construction of
the precursors themselves) should be selected in a way to optimize
the predictive power of the forecasting algorithm. The optimization
procedure will be presented below, its goal is to adapt the parameters
of precursors to a given catalog of earthquakes, that is to obtain
the best possible retrospective forecast. However, this adaptation
procedure creates a \textquotedbl{}ghost\textquotedbl{} information
related with the specific features of the given catalog but not present
in physical propertities of real seismicity. This \textquotedbl{}ghost\textquotedbl{}
information will not be reproduced if the algorithm is applied to
another catalog of earthquakes. It is necessary to increase the volume
of the catalog and to reduce the number of free parameters to get
rid of this \textquotedbl{}ghost\textquotedbl{} information.. Clearly,
the first goal requires the considerable increase of the observation
period and may be achieved in the remote future only. So, one concentrates
on the reduction of number of degrees of freedom. The simplest ansatz
including $Q-1$ parameters is the linear combination

\begin{equation}
f^{*}=Ff_{\text{1}}+\sum_{q=2}^{Q}c_{q-1}F\mathit{\mathit{f}_{q}}.\label{eq:3}
\end{equation}
As a strictly monotonic function of precursor is a precursor itself
the log-linear combination is an equally suitable choice

\begin{equation}
f^{*}=\ln\left(Ff_{\text{1}}\right)+\sum_{q=2}^{Q}c_{q-1}\ln\left(F\mathit{\mathit{f}_{q}}\right)\text{,}\label{eq:4}
\end{equation}
Here $c_{q},\, q=1,...,Q-1$ are free parameters. The result of the
procedure has the form $g=Ff^{*}.$

\section{Alarm levels, point and interval estimations }

In view of (\ref{eq:1}) the precursor $g$ is the set of estimates
for probabilities 
\[
\mathrm{Pr}\left\{ \mathit{h=\mathrm{1\mid}z_{l-1}^{*}\leq f(i,j,k)<z_{l}^{*}}\right\} ,\, l=1,...,L(f).
\]
Its serious drawback is that typically $\left\{ z_{l-1}^{*}\leq f(i,j,k)<z_{l}^{*}\right\} $
correspond to single events, and therefore the precision of these
estimates is very low (the confidence intervals discussed below may
be taken as a convenient measure of precision). In order to increase
the precision it is recommended to use the larger cells containing
a larger number of events, that is a more coarse covering of the space
where the precursor takes its values. In a sense, the precision of
the estimation and the localization of the precursor values in its
time-space region are related by a kind of \textquotedbl{}uncertainty
principle\textquotedbl{}: the more precise estimate one wants to get
the more coarse is the time-space range of their values and vice versa.

We adapt the following approach in order to achieve a reasonable compromise.

1. For fixed thresholds $a_{s},\, s=1,...,S+1,\: a_{1}=1,\: a_{s}<a_{s+1},\: a_{S+1}=0,$
we define $\text{\ensuremath{\mathit{S}}}$ possible alarm levels
$a_{s+1}\leq g(i,j,k)<a_{s}$ and subsets $\Omega_{s},\, s=1,...,S$,
of the set $\Omega_{K}$ corresponding to alarm levels, i.e., $\Omega_{s}$
is a set of cells of $\Omega_{K}$, such that $a_{s+1}\leq g(i,j,k)<a_{s}$

There are different ways to choose the number $S$ of alarm levels
and the thresholds $a_{s},\, s=2,...,S$. Say, fix $S=5$, and select
$a_{s}=10^{-\alpha(s-1)}$. This is a natural choice of the alarm
level because at $\alpha=1$ it corresponds to decimal places of the
estimate of the conditional probability given by the precursor. The
problem with $S=2$, i.e., two-level alarm, may be reduced to the
hypothesis testing and discussed in more details below.

2. Compute the point estimates $\theta_{\text{s}}$ of probabilities
$\mathrm{Pr}\left\{ \mathit{h=\mathrm{1\mid}a_{s+\mathrm{1}}\leq g\mathrm{(}i,j,k\mathrm{)}<a_{s}}\right\} ,$
$s=1,...,S$, obtained via the distribution $P_{\Omega}(x)$ of precursor
$g$ in the same way as in (\ref{eq:2}). The property (iv) implies
that for any domain $\Omega_{s}$ the binary random variables $h\mathrm{(}i,j,k\mathrm{)}$
are independent and identically distributed, i.e

\[
\mathrm{Pr}\left\{ \mathit{h=\mathrm{1\mid}a_{s+1}\leq g(i,j,k)<a_{s}}\right\} \equiv p_{s},
\]
and the unbiased estimate of $p_{s}$ takes the form

\begin{equation}
\theta_{s}=\frac{m_{s}}{n_{s}}\label{eq:9}
\end{equation}
where $n_{s}$ stands for the number of cells in domain $\Omega_{s}$,
and by $m_{s}$ we denote the number of cells in $\Omega_{s}$ containing
seismic events.

3. The random variable $m_{s}$ takes integer values between $0$
and $n_{s}$. The probabilities of these values are computed via the
well-known Bernoulli formula $\mathrm{Pr}(m_{s}=k)=\binom{n_{s}}{k}\, p_{s}^{k}(1-p_{s})^{n_{s}-k}$.
Let us specify the confidence interval covering the unknown parameter
$p_{s}$ with the confidence level $\gamma$. In view of the integral
Mouvre-Laplace theorem for $n_{s}$ large enough the statistics $\frac{(\theta_{s}-p_{s})\sqrt{n_{s}}}{\sqrt{p_{s}(1-p_{s})}}$
is approximately Gaussian N$(0,1)$ with zero mean and unit variance.
Note that the values $n_{s}$ increase with time. Omitting straightforward
calculations and replacing the parameter $p_{s}$ by its estimate
$\theta_{s}$ we obtain that $\theta_{s}^{-}<\theta_{s}\text{<\ensuremath{\theta_{s}^{+}}},$
where $\theta_{s}^{-}=\theta_{s}-\frac{t_{\gamma}\sqrt{\theta_{s}(1-\theta_{s})}}{\sqrt{n_{s}}}$,
$\theta_{s}^{+}=\theta_{s}+\frac{t_{\gamma}\sqrt{\theta_{s}(1-\theta_{s})}}{\sqrt{n_{s}}}$,
and $t_{\gamma}$ is the solution of equation $\Phi(t_{\gamma})=\frac{\gamma}{2}$.
Here $\Phi$ stands for the standard Gaussian distribution function.

4. As a result of these considerations we introduce 'the precursor
of alarms' which indicates the alarm level: $\mathbf{R}(\mathbf{f}(i,j,k))=s(i,j,k)$.
It will be used for calculations of point estimate and the confidence
inteval in the form $\{\theta_{s(i,j,k)}^{-}<\theta_{s(i,j,k)}<\theta_{s(i,j,k)}^{+}\}$.
This result will be use for prospective forecasting procedure.

\section{The information gain and the precursor quality}

The construction of a 'combined' precursor $\mathbf{R}$ involves
parameters from formula (\ref{eq:3}) or (\ref{eq:4}) as well as
parameters which appear in definition of each individual precursors
$f_{q}$. It is natural to optimize the forecasting algorithm in such
a way that the information gain related to the seismic events is maximal.
In one-dimensional case the information gain as a measure of the forecast
efficiency was first intoduced by Vere-Jones {[}Vere-Jones, 1998{]}.
Here we exploit his ideas in the case of multidimensional space-time
process.

Remind the notions of the entropy and information. Putting aside the
mathematical subtlety (see {[}Kelbert, Suhov, 2013{]} for details)
we follow below an intuitive approach of the book {[}Prohorov, Rozanov,
1969{]}. The information containing in a given text is, basically,
the length of the shortest compression of this text without the loss
of its content. The smallest length $S$ of the sequence of digits
$0$ and $1$ (in a binary code) for counting $N$ different objects
satisfies the relations $0\leq S-\log_{2}N\leq1$. So, the quantity
$S\approx\log_{2}N$ characterizes the shortest length of coding the
numbers of $N$ objects.

Consider an experiment that can produce one of $N$ non-intersecting
events $\mathit{\lyxmathsym{\textcyr{\char192}}}_{1},\lyxmathsym{\ldots},\mathit{\lyxmathsym{\textcyr{\char192}}_{N}}$
with probabilities $\mathit{q}_{1},\lyxmathsym{\ldots},\mathit{q}_{N}$,
respectively, $\mathit{q}_{1}+\lyxmathsym{\ldots}+\mathit{q_{N}}=1$.
A message informing about the outcomes of $n$ such independent identical
experiments may look as a sequence $(A_{i_{1}},\lyxmathsym{\ldots},A_{i_{n}})$,
where $A_{i_{k}}$ is the outcome of the experiment $k$. But for
long enough series of observations the frequency $n_{i}/n$ of event
$\mathit{\lyxmathsym{\textcyr{\char192}}}_{i}$ is very close to its
probability $\mathit{q}_{i}$. It means that in our message $(A_{i_{1}},\text{\ensuremath{\ldots}},A_{i_{n}})$
the event $\mathit{\lyxmathsym{\textcyr{\char192}}}_{i}$ appears
$n_{i}$ times. The number of such messages is

\[
N_{n}=\frac{n!}{n_{1}!...n_{N}!}.
\]
By the Stirling formula the length of the shortest coding of these
messages

\[
S_{n}\approx\log_{2}N_{n}\approx-n\sum_{i=1}^{N}q_{i}\log_{2}q_{i}.
\]
The quantity $S_{n}$ measures the uncertainty of the given experiment
before its start, in our case we are looking for one of possible outcomes
of $n$ independent trials. The specific measure of uncertainty for
one trial

\[
\frac{1}{n}S_{n}=\frac{1}{n}S_{n}(\mathit{q_{\mathrm{1}},\lyxmathsym{\ldots},q_{N}})=-\sum_{i=1}^{N}q_{i}\log_{2}q_{i}
\]
is known as Shannon's entropy of distribution $\mathit{q\mathrm{_{1}},\lyxmathsym{\ldots},q_{N}}$
(in physical literature it is also known as a measure of chaos or
disorder). After one trial the uncertainty about the future outcomes
decreases by the value $S=S_{n}-S_{n-1}$, this decrement equals to
the \emph{information gain} $I=S$, obtained as a result of single
trial.

The quantity

\begin{equation}
S(h)=-p_{A}\log_{2}p_{A}-(1-p_{A})\log_{2}(1-p_{A})\label{eq:6}
\end{equation}
is the (unconditional) entropy of distribution for indicator of seismic
event $h$ in a space-time cell in the absence of any precursors.
The conditional entropy $S(h\mid a_{s+1}\leq g<a_{s})$ under condition
that in the cell $(i,j,k)$ the alarm level $s$ is set up equals

\[
S(h\mid a_{s+1}\leq g<a_{s})=-p_{s}\log_{2}p_{s}-(1-p_{s})\log_{2}(1-p_{s})
\]
The expected conditional entropy $S_{\boldsymbol{\mathbf{R}}}(h)$
of indicator of seismic events where the averaging in taken by the
distribution of precursors $\mathrm{\mathbf{R}}$ takes the form

\begin{equation}
S_{\boldsymbol{\mathbf{R}}}(h)=-\sum_{s=1}^{S}\left[p_{s}\log_{2}p_{s}+(1-p_{s})\log_{2}(1-p_{s})\right]P_{A}(a_{s+1}\leq g<a_{s})\label{eq:5}
\end{equation}

We conclude that the knowledge of the precursor values helps to reduce
the uncertainty about the future experiment by $S(h)-S_{\mathrm{\mathbf{R}}}(h)$
which is precisely information $I(\mathrm{\mathbf{R}},h)$ obtained
from the precursor. Taking into account (\ref{eq:6}), (\ref{eq:5})
and the fact that

\[
p_{A}=\sum_{s=1}^{S}p_{s}P_{A}(a_{s+1}\leq g(i,j,k)<a_{s})
\]
we specify the information gain as

\begin{gather*}
I(\mathrm{\mathbf{R}},h)=\sum_{s=1}^{S}\left[p_{s}\log_{2}\frac{p_{s}}{p_{A}}+(1-p_{s})\log_{2}\frac{1-p_{s}}{1-p_{A}}\right]P_{A}(a_{s+1}\leq g<a_{s}).\\
\end{gather*}

By analogy with the one-dimensional case {[}Kolmogorov, 1965{]} the
quantity $I(\mathrm{\mathbf{R}},h)$ may the called the mutual information
about the random field $h$ that may be obtained from observations
of random field $\mathbf{R}$. It is known that the information $I(\mathrm{\mathbf{R}},h)$
is non-negative and equals to $0$ if and only if the random fields
$h$ and $\mathbf{R}$ are independent. This mutual information $I(\mathrm{\mathbf{R}},h)$
takes its maximal value $S(h)$ in an idealized case of absolutely
exact forecast. The mutual information quantifies the information
that the distributions of precursors contribute to that of the indicator
of seismic event. For this reason it may be considered as an adequate
scalar estimate for the quality of the forecast.

The quantity $I(\mathrm{\mathbf{R}},h)$ depends on the cell size,
i.e., on the space discretization length $a$ and time interval $\Delta t$.
We need a formal test to compare precursors defined for different
size of the discretization cells. For this aim let us introduce the
so-called 'efficiency' of precursors as the ratio of information gains

\[
r(\mathrm{\mathbf{R}},h)=\frac{I(\mathrm{\mathbf{R}},h)}{S(h)}.
\]
This efficiency varies between $0$ and $1$ and serves as a natural
estimate of information quality of precursors. It allows to compare
different forecasting algorithms and select the best one.

A natural estimate of $S(h)$ based on (\ref{eq:7}) and (\ref{eq:8})
is defined as follows

\begin{equation}
\hat{S}(h)=-\lambda\log_{2}\lambda-(1-\lambda)\log_{2}(1-\lambda).\label{eq:12}
\end{equation}
Taking into account (\ref{eq:9}) and using an estimate of $P_{A}(a_{s+1}\leq g<a_{s})$
in the form of ratio $\tau_{s}=\frac{n_{s}}{N},$ we construct an
estimate of $I(\mathrm{\mathbf{R}},h)$ as follows

\begin{equation}
\hat{I}(\mathrm{\mathbf{R}},h)=\sum_{s=1}^{S}\left[\theta_{s}\log_{2}\frac{\theta_{s}}{\lambda}+(1-\theta_{s})\log_{2}\frac{1-\theta_{s}}{1-\lambda}\right]\tau_{s},\label{eq:13}
\end{equation}

\begin{equation}
\hat{r}(\mathrm{\mathbf{R}},h)=\frac{\hat{I}(\mathrm{\mathbf{R}},h)}{\hat{S}(h)}.\label{eq:14}
\end{equation}

\begin{remark}
\textit{The economical quality of forecast.} A natural economic measure
for a quality of binary forecast is the economic risk or damage $r$
related to the earthquakes and the necessary protective measures.
In mathematical statistics the risk is defined as the expectation
of the loss function, in our case there are two types of losses: damage
and expenses related to protection. For each cell of our grid the
risk may be specified by the formula 
\[
\begin{array}{l}
r=\alpha\mathrm{Pr}\text{\{}\mathit{h}(\mathit{i},\mathit{j},\mathit{k})=1,\eta(\mathit{i},\mathit{j},\mathit{k})=0\}+\beta\mathrm{Pr}\{\mathit{h}(\mathit{i},\mathit{j},\mathit{k})=0,\eta(\mathit{i},\mathit{j},\mathit{k})=1\}+\\
+\gamma\mathrm{Pr}\{\mathit{h}(\mathit{i},\mathit{j},\mathit{k})=1,\eta(\mathit{i},\mathit{j},\mathit{k})=1\},
\end{array}
\]
here $\alpha$ stands for the average damage from a seismic event;
$\beta$ stands for the average expenses for protection after a seismic
alarm is announced; $\gamma$ stands for the average damage after
the alarm, $\gamma=\alpha+\beta-\delta$, where $\delta$ is the damage
prevented by the alarm. The coefficient in front of $\mathrm{Pr}\{\mathit{h}(\mathit{i},\mathit{j},\mathit{k})=0,\eta(\mathit{i},\mathit{j},\mathit{k})=0\}$,
obviously, equals $0$, because in the absence both of a seismic event
and an alarm there is no loss of any kind. Clearly, only the case
when $\delta>\beta$ is economically justified, i.e., the gain from
the prevention measures is positive. Obviously, $\delta$ should be
less than $\alpha+\beta$, i.e., an earthquake cannot be profitable.
Taking into account that $\alpha,\beta$ and $\gamma$ depend on the
geographical position of the cell, we write the total risk as the
summation over all cells in the region of a given forecast. In the
simplest case of the absence of the spacial component, when a single
cell represents a region of forecast, the expression for the risk
is simplified as follows $r=\alpha\lambda\nu+\beta\tau+\gamma\lambda(1-\nu).$
\end{remark}
However, the risk $r$, which is very useful for economical considerations
and as a basis for an administrative decision, could hardly be used
as a criteria for quality of seismic prediction. First of all, it
cannot be computed in a consistent way because the coefficients $\alpha,\:\beta$
and $\gamma$ are not known in practice, and hence no effective way
of its numerical evaluation is known. The computation of these coefficients
is a difficult economic problem and goes far beyond of the competence
of geophysicists. On the other hand, the readiness of the authority
to commit resources to solving the problem depends on the quality
of the geophysical forecast. This situation leads to a vicious cirle.

The second drawback of the economic risk as a criterion for the quality
of prediction is related to the fact that it depends on many factors
which have no relation to geophysics or earthquake prediction. These
factors inlude the density of population, the number and size of industrial
enterprises, infrastructure, etc. It also depends on subjective factors
such as the williness of authorities to use resources for prevention
of the damage from earthquakes. The natural sciences could hardly
accept the criteria for the forecast quality which depend on the type
of state organization, priorities of ruling parties, results of the
recent elections, etc. 
\begin{remark}
It seems reasonable to introduce a penalty related to the number of
superfluous parameters in evaluating the quality of forecast pointing
to the natural analogy with the Akaike test {[}Akaike, 1974{]} and
similar methods in information theory. In our context the main parameter
of importance is ${\hat{r}}({\bf R},t)$ and its limit as $t\to\infty$.
This quantity does not involve the number of parameters directly.
Probably, the rate of convergence depends on the number of parameters
but this dependence is not studied yet. 
\end{remark}

\section{The forecasting procedure }

The number of time intervals, i.e., the number of observation $N$
used in the construction of estimates increases with the growth of
observation time. So, the computation procedure requires constant
innovations. On the other hand some computation time is required to
'adapt' the model parameters to the updated information about seismic
events via an iterative procedure. For these reasons we propose the
following forecasting algorithm.

1. Given initial parameter values at the moment $t_{K-1}=t_{0}+(K-2)\Delta t$
we optimize them to obtain the maximum of efficiency $\hat{r}(\mathrm{\mathbf{R}},h)$
of precursor in domain $\Omega_{K-1}.$ For this aim the Monte-Carlo
methods is helpful: one perturbs the current values of parameters
randomly and adapts the new values if the efficiency increases. The
process continues before the value of efficiency stabilized, this
may give a local maximum, so the precedure is repeated sufficient
number of times. The choice of initial value on the first step of
optimization procedure is somewhat arbitrary but a reasonable iteration
procedure usually leads to consistent results. The opmization procedure
takes the period of time $t_{K-1}<t\leq t_{K}$ .

2. Next, we construct the forecast in the following way. At the moment
$t_{K}$ the values of precursor $\hat{g}$ in each cell $(i,j,K+1)$
is computed with optimized parameters. Based on these parameter values
the alarm levels, the point estimates and confidence intervals are
computed in each cell as well as the values of efficiency of precursors.

3. The estimates of stationary probabilities of seismic events in
the cell $\bar{\theta}(i,j)$ are defined as follows:

\[
\bar{\theta}(i,j)=\frac{1}{K}\sum_{k=1}^{K}\theta_{s(i,j,k)}.
\]
they can be used for creation of the of\textcolor{black}{{} the variant
of the m}aps of seismic hazard in the region.

\section{Retrospective and prospective informativities}

The efficiency of precursor which is achieved as a result of parameters
optimization could be considered as \emph{retrospective} as it is
constructed by the precursors adaptation to the historical catalogs
of seismic events. The prospective efficiency for the space-time domain
$\Omega^{*}$ containing the cell in the 'future' is based on the
forecast. It is computed via formulas (\ref{eq:12}), (\ref{eq:13}),
(\ref{eq:14}) with the only difference that domain $\Omega_{\mathit{s}}^{*}$
consists from the cells where the forecasted alarm level is $s$.
The efficiency of prospective forcast is smaller compared with the
retrospective efficiency, however approaches this value with time.
In principle, the prospective efficiency is an ultimate criteria of
precursors quality and the retrospective efficiency could serve only
for the preliminary selection of precursors and their adaptation to
the past history of seismic events.

\section{Testing of the forecasting algorithm }

The efficiency of precursor could be computed exactly only in an idealized
case of infinite observation time. However, its estimate may be obtained
based on the observation over a finite time interval. So, if an estimate
produces a non-zero value not necessarily the real effects is present.
It may be simply a random fluctuation even if the precursor provides
no information about the future earthquake. For this reason we would
like to check the hypothesis $H_{0}$ about the independence of a
precursor and an event indicator with a reasonable level of confidence.
In case the hypothesis is rejected one have additional assurance that
the forecasting is real, not just a \textquotedbl{}ghost\textquotedbl{}.

So, consider the distributions

\[
\mathit{P}_{A}(\mathit{x})=\frac{1}{\mathit{I\cdot J}}\sum_{\left(\mathit{i,j}\right)\in\Omega}\mathrm{P}\mathrm{r}_{ij}\left\{ \mathit{g\mathrm{(}i,j,k\mathrm{)\leq\mathit{x}}}\right\} ,
\]
and 
\[
\mathit{P_{A}^{\prime}}(\mathit{x})=\frac{1}{\mathit{I\cdot J}}\sum_{\left(\mathit{i,j}\right)\in A}\mathrm{P}\mathrm{r}_{ij}\left\{ \mathit{g\mathrm{(}i,j,k\mathrm{)\leq\mathit{x}}\mid}h\mathrm{(}i,j,k\mathrm{)=1}\right\} 
\]
The function $\mathit{P}_{A}(\mathit{P}_{A}^{-1}(y))=y$ of variable
$y=\mathit{P}_{A}(x)$ provides an uniform distribution $F^{*}(y)=\mathrm{Pr}\big({\xi\leq y}\big)$
of some random variable $\xi$ on {[}0,1{]}. Next, consider a distribution
function $G(y)=P_{A}^{\prime}(\mathit{P}_{A}^{-\mathrm{1}}(y))$ on
{[}0,1{]}, and use a parametric representation for abcissa $\mathit{P}_{A}(\mathit{x})$
and ordinate $\mathit{P}_{A}^{\prime}(\mathit{x})$. If random fields
$g$ and $h$ are independent the distribution functions $\mathit{P}_{A}(\mathit{x})=\mathit{P}_{A}^{\prime}(\mathit{x})$
and $G(y)$ are uniform. So, the hypothesis about the absence of forecasting,
i.e., about the independence of $g$ and $h$, is equivalent to the
hypothesis $H_{0}$ that the distribution $G(y)$ is uniform.

The empirical distribution $G_{L}(y)$ related to $G(y)$ is defined
as follows. Denote by $u_{l},l=1,...L$ the values of the function
$g(i,j,k)$ sorted in the non-decreasing order and beloning to the
cells where $h\mathrm{(}i,j,k\mathrm{)=1.}$ Let $n_{l}$ be the numbers
of cells such that $h\mathrm{(}i,j,k\mathrm{)=1,}\, g(i,j,k)=u_{l}.$
Denote by $m(u_{l})$ the numbers of cells from $\Omega$ such that
$g(i,j,k)<u_{l}$, and define the empirical distribution $G_{L}(y)$
as a step-wise function with $G_{L}(0)=0$ and positive jumps of the
size $\frac{n_{l}}{L}$ at points $y_{l}=\frac{m(u_{l})}{N}$.

The well-known methods of hypothesis testing requires that the function
$G_{L}(y)$ has the same shape as for independent trials, i.e., random
variables $u_{l},l=1,...L$ are independent in view of axiom (iv).
Naturally, we accept the precursors such that the hypothesis $H_{0}$
is rejected with the reasonable level of confidence. (Remind, that
the hypothesis is accepted if and only if its logical negation could
be rejected based on the available observations. The fact that the
hypothesis cannot be rejected does not mean at all that it should
be accepted, it only means that the available observations don't contradict
this hypothesis. Say, the well-known fact that \textquotedbl{}The
Sun rise in the East\textquotedbl{} does not contradict to our hypothesis,
however it may not be considered as a ground for its acceptance.)
For large values of $L$ the Kolmogorov statistics {[}Kolmogorov,
1933a{]} is helpful for this aim

\[
D_{L}=\sup\mid G_{L}(y)-y\mid
\]
with an asymptotic distribution

\[
\underset{L\rightarrow\infty}{\lim}\mathrm{Pr\left\{ \mathit{\sqrt{L}D_{L}\leq z}\right\} =\sum_{\mathit{k}=-\infty}^{\infty}}\left(-1\right)^{k}\mathit{e^{-2k^{2}z^{2}},\, z>\mathrm{0},}
\]
or Smirnov's statistics {[}Smirnov, 1939{]}

\[
D_{L}^{+}=\sup\left[G_{L}(y)-y\right],
\]

\[
D_{L}^{-}=-\inf\left[G_{L}(y)-y\right],
\]
with asymptotic distribution

\[
\underset{L\rightarrow\infty}{\lim}\mathrm{Pr\left\{ \mathit{\sqrt{L}D_{L}^{+}\leq z}\right\} =\underset{\mathit{L}\rightarrow\infty}{\lim}\mathrm{Pr\left\{ \mathit{\sqrt{L}D_{L}^{-}\leq z}\right\} =1-\mathit{e^{-2z^{2}},\, z>\mathrm{0.}}}}
\]
The asymptotic expressions for these statistics can be used for $L>20$
({[}Bolshev, Smirnov, 1965{]})..

\section{The binary alarm and the hypothesis testing}

The prediction is the form of forecast when an alarm is announced
in a given cell without a preliminary evaluation of probability of
seismic event. In this case we can estimate the probabilities of events
too. (If the alarm is announced in an arbitrary domain $\Omega$ we
set up an alarm if at least \textcolor{black}{haph of the cell of
our model is occupied by al}arm.).

Let $M$ be the number of cells in $\Omega$ which are in the state
of alarm, $M_{0}$ be the number of cells where the seismic event
is present but no alarm was announced (the number of 'missed targets').
Denote by $\tau=\frac{M}{N}$ the share of the cells with alarm announced,
$\lambda=\frac{L}{N}$ the share of the cells with seismic events,
and $\nu=\frac{M_{0}}{M}$ the share of missed targets. Let a random
variable $\eta(i,j,k),$ equal $1$ if an alarm is announced in the
cell $(i,j,k)$, and $0$ otherwise. Obviously, the estimate of conditional
probability $\mathrm{Pr\left\{ \mathit{h}(\mathit{i},\mathit{j},\mathit{k})=1\mid\eta(\mathit{i},\mathit{j},\mathit{k})=1\right\} }$
of the seismic event under the condition of alarm is $\frac{\lambda(1-\nu)}{\tau},$
and the estimate of conditional probability $\mathrm{Pr\left\{ \mathit{h}(\mathit{i},\mathit{j},\mathit{k})=1\mid\eta(\mathit{i},\mathit{j},\mathit{k})=0\right\} }$
of the seismic event under the condition of no alarm is $\frac{\lambda\nu}{1-\tau}.$

If the alarm is announced according to the procedure described in
Section 5 the threshold $a_{1}$ specifying the acceptable domain
of values for $g(i,j,k)$ should be treated as a free parameter and
selected by maximizing the information efficiency $\hat{r}(\eta,h)$.
The estimate of information increase for given values of $\tau$ and
$\nu$ equals

\begin{gather*}
\hat{I}(\eta,h)=\lambda(1-\nu)\log_{2}\frac{1-\nu}{\tau}+\lambda\nu\log_{2}\frac{\nu}{1-\tau}+\\
+\left[\tau-\lambda(1-\nu)\right]\log_{2}\frac{\tau-\lambda(1-\nu)}{(1-\lambda)\tau}+\left(1-\tau-\lambda\nu\right)\log_{2}\frac{1-\tau-\lambda\nu}{(1-\lambda)(1-\tau)}.
\end{gather*}

The value of $\eta(i,j,k)$ characterizes the results of checking
two mutually exclusive simple hypothesis:

$H_{0}$: the distribution of $\tilde{g}\mathrm{(}i,j,k\mathrm{)}$
has the form $\mathrm{\mathit{P}_{\mathit{A}}^{0}(\mathit{x})\equiv}$
$P\mathrm{r}_{\mathit{A}}\left\{ \mathit{\tilde{g}\mathrm{(}i,j,k\mathrm{)\leq\mathit{x}\mid\mathit{h}\mathit{\mathrm{(}i\mathrm{,}j\mathrm{,}k}\mathrm{)=0}}}\right\} $,
implying 'no seismic events',\\
 or

$H_{1}:$ the distribution of $\tilde{g}\mathrm{(}i,j,k\mathrm{)}$
has the form $\mathrm{\mathit{P}_{\mathit{A}}^{1}(\mathit{x})\equiv}$
$P\mathrm{r}_{\mathit{A}}\left\{ \mathit{\tilde{g}\mathrm{(}i,j,k\mathrm{)}\mathrm{\leq\mathit{x}\mid\mathit{h}\mathit{\mathrm{(}i\mathrm{,}j\mathrm{,}k}\mathrm{)=1}}}\right\} $,
implying the presence of seismic event.\\

Statistics for checking of these hypothesis is the precursor $g\mathrm{(}i,j,k\mathrm{)}$,
and the critical domain for $H_{0}$ has the form $\left\{ g\mathrm{(}i,j,k\mathrm{)}\geq a_{1}\right\} $.
(If usual method of alarm announcement is used the relevant precursor
plays the rôle of statistics and the critical domain is defined by
the rule of the alarm announcement). The probability of first type
error 
\[
\alpha=\mathrm{Pr\left\{ \mathit{\eta}(\mathit{i},\mathit{j},\mathit{k})=1\mid\mathit{h}(\mathit{i},\mathit{j},\mathit{k})=0\right\} ,}
\]
it is estimated as $\frac{\tau-\lambda(1-\nu)}{1-\lambda}$. The probability
of second type error 
\[
\beta=\mathrm{Pr\left\{ \mathit{\eta}(\mathit{i},\mathit{j},\mathit{k})=0\mid\mathit{h}(\mathit{i},\mathit{j},\mathit{k})=1\right\} ,}
\]
it is estimated as $\nu.$ (Note that due to condition (iii) any test
used for the checking these hypothesis should not depend on the coordinates
of the cell).

The Neyman-Pearson theory allows to define the domain \textcolor{black}{of
images}\textcolor{red}{{} }of all possible criteriaall possible criteria:
in coordinates $(\alpha,\beta)$ it is a convex domain with a boundary
$\Gamma$ which corresponds to the set of uniformly most powerful
tests. This family may be defined in terms of the likelihood ratio
$\Lambda(x)=\frac{\mathit{p}_{\mathit{A}}^{1}(\mathit{x})}{\mathit{p}_{\mathit{A}}^{0}(\mathit{x})}$
under condition that the distributions $\mathit{P}_{\mathit{A}}^{1}(\mathit{x})$
and $\mathit{P}_{\mathit{A}}^{0}(\mathit{x})$ has densities $\mathit{p}_{\mathit{A}}^{1}(\mathit{x})$
and $\mathit{p}_{\mathit{A}}^{0}(\mathit{x})$: 
\[
\begin{cases}
\eta(i,j,k\mathrm{)}=1\:\mathrm{if\:}\Lambda(x)>\omega,\\
\eta(i,j,k\mathrm{)}=0\:\mathrm{if\:}\Lambda(x)<\omega
\end{cases}
\]
where $\omega$ denotes the threshold. In the paper {[}Gercsik, 2004{]}
we demonstrated that among all the tests with the images on the boundary
$\Gamma$ there exists three different best tests. Here the term \textquotedbl{}best\textquotedbl{}
may be understood in three different sense, i.e., maximizing the variational,
correlational and informational efficiency. The most relevant criteria
is the informational efficiency ${\hat{r}}(\eta,h)$.

The well-known Molchan's error diagram {[}Molchan, 1990{]} where the
probability of the first kind error is estimated by $\tau$ is constructed
in the same way. However, it involve a comparison of two intersecting
hypothesis:

$H_{0}$: the distribution $\tilde{g}\mathrm{(}i,j,k\mathrm{)}$ has
the form $\mathrm{\mathit{P}_{\mathit{A}}^{0}(\mathit{x})\equiv P}\mathrm{r}_{\mathit{A}}\left\{ \mathit{\tilde{g}\mathrm{(}i,j,k\mathrm{)}\mathrm{\leq\mathit{x}}}\right\} $,
i.e., the seismic event could \textquotedbl{}either happen or not
happen\textquotedbl{}, and

$H_{1}:$ the distribution $\tilde{g}\mathrm{(}i,j,k\mathrm{)}$ has
the form $\mathrm{\mathit{P}_{\mathit{A}}^{1}(\mathit{x})\equiv P}\mathrm{r}_{\mathit{A}}\left\{ \mathit{\tilde{g}\mathrm{(}i,j,k\mathrm{)}\leq\mathit{x}\mid\mathit{h}\mathit{\mathrm{(}i\mathrm{,}j\mathrm{,}k}\mathrm{)=1}}\right\} $,
i.e., the seismic event \textquotedbl{}will happen\textquotedbl{}

\noindent .Note that the rejection of hypothesis $H_{0}$ leads to
absurd results.
\begin{remark}
\noindent In the paper {[}Molchan and Keilis-Borok, 2008{]} the area
of the alarm domain is defined in terms of non-homogeneous measure
depending on the spacial coordinates, in terms of our paper it may
be denoted as ${\bar{\theta}(i,j)}$. i.e., $\tau\backsim\sum_{i,j}\bar{\theta}(i,j)\eta(i,j,k\mathrm{)}$.
This approach is used to eliminate the decrease of the share of alarmed
sites $\tau$ with the extension of the domain when a purely safe
and aseismic territory is included into consideration. It would be
well-justified if the quantity $\tau$ could be accepted as an adequate
criterion of the quality of forecast in its own right. On the other
hand, it can be demonstrated that the information efficiency ${\hat{r}}(\eta,h)$
converges to a non-zero value $1-\nu$ when the number of cells with
an alarm is fixed but the total number of cells tends to infinity.
An inhomogeneous area of the territory under forecast which is proportional
to $\bar{\theta}(i,j)$ does not enable us to calculate the informational
efficiency. Moreover, it possesses a number of unnatural features
from the point of view of evaluation of economical damage. A seismic
event in the territory of low seismicity is more costly because no
precautions are taken to prevent the damage of infrastructure. However,
in this inhomogeneous area an alarm announced in an aseismic territory
will have a smaller contribution than an alarm in a seismically active
territory where the losses would be in fact smaller. We conclude that
this approach 'hides' the most costly events and does not provide
a reasonable estimate of economic damage.
\end{remark}

\section{The choice of precursors}

We use the term 'empiric precursor of earthquake' for any observable
characterisric derived from the catalog\textcolor{black}{{} only which
provides for this catalog a reasonable retrospective forcast of seismic
events and not derived from basic physical conception of seismicity
(say, the periods of relatice calm, deviation of some basic characteristic
from a long-time average , etc). In contrast, the physical precursors
are de- rived from some of physical processe and characterize physical
quantities (stress fields, strength, concentration of cracks, etc.)
or well-defined physical processes (}i.e., phase transitions, cracks
propagations, etc.) In the meteorological forecast the danger of using
empirical precursors was highlighted by A. Kolmogorov in 1933 {[}Kolmogorov,
1933{]}. From that time the meteorological forecast relies on the
physical precursors which are theoretically justified by the models
of atmospheric dynamics. Below we will present A. Kolmogorov's argument
adapted to the case of seismic forecast. This demonstrates that the
purely empirical precursors work well only for the given catalog from
which they are derived. However, their eficiency deteriorates drastically
when they are applied to any other independent catalog. 

Consider a group of $k$ empirical precursors used for a forecast
and and selected from a set of $n$ such groups. According to A. Kolmogorov's
remark the number $k$ is typically rather small. This is related
to the fact that a number of strong earthquakes in catalog is unlikely
to exceed a few dozen. As the values of precursors are random there
exists a small probability $p$ that the efficiency of the forecast
exceeds the given threshold $\lyxmathsym{\textcyr{\char209}}$. Then
the probability of event $\hat{r}(\mathrm{\mathbf{R}},h)\text{\ensuremath{\le}}\text{\textcyr{\char209}}$
equals $1\lyxmathsym{\textendash}p$, and the probability of event
$\hat{r}(\mathrm{\mathbf{R}},h)>\text{\textcyr{\char209}}$ for at
least one collection of precursors equals $P=1-(1-p)^{n}$ and tends
to $1$ as $n\rightarrow\infty$.. (According to Kolmogorov some arbitrariness
of the assumption of independence is compensated by the large number
of collections.)

Summing up, if the number of groups is large enough with probability
close to 1 it is possible to find a group giving an effective retrospective
forecast for a given catalog. In practice this is always the case
as the number of empirical precursors could be increased indefinitely
by variation of real parameters used in their construction. It is
important to note that for such a group, which is highly eficient
for the initial catalog, the probability that the eficiency is greater
than $C$ is still equal to $p$ for any other catalog. In other words
the larger the number of the groups of empirical precursors the less
reliable forecast is. So, the collection of a large list of the empirical
precursors is counter-productive. 

Much more reliable are the physical precursors intrinsicly connected
with the physical processes which preserve their values with the change
of sample. The probability to find such a set of precursors by pure
empirical choice is negligible because they are very rare in the immense
collection of all possible precursors.

\section{Image identification}

The possibility to use the pattern recognition formalism in seismic
forecast is totally based on the acceptance of deterministic model
of seismicity. It is necessary to assume that in principle there exits
such a group of precursors which allows to determine with certainly
whether a strong earthquake will happen or not. In this case one believes
that all random errors are related to the incompleteness of this set
of precursors.But if the seismicity is a random process then the image
appears only after the earthquake and before it any set of values
for precursors cannot guarantee the possible outcome and only the
relevant probabilities may be a subject of scientific study. After
the discovery of dynamic instability and generators of stochastic
behavior of dynamical systems the deterministic model of seismicity
is cast in doubts. Its potential acceptance requires substantial evidence
which hardly exist at present.

In any case the results of pattern recognition procedure (i.e., a
binary alarm) are useful if they are considered alongside with the
results of statistical tests. They allows to calculate the estimate
of probabilities of seismic events and informational efficiency.

However, the section of 'features' for pattern recognition leads to
the same difficulties as the selection of precursors: the 'features'
based on the observations only and not related to the physics of earthquakes
are not helpful, and any hopes for 'perceptron education' are not
grounded. A successful supervised recognition is possible if the features
has proved causal relation with pattern. This principle is illustrated
by a simple but important theorem by A.N. Krichevets.
\begin{theorem}
Let $A$ be a finite set, $B_{1},B_{2}\subset A,B_{1}\cap B_{2}=\emptyset$.
We say that $B_{1}$ and $B_{2}$ are finite educational samples.
Let $X\in A$, $X\notin B_{1}\cup B_{2}$ be a new object. Then among
all classifications, i.e., subsets $(A_{1},A_{2})$ such that $B_{1}\subset A_{1}$,
$B_{2}\subset A_{2}$, $A_{1}\cup A_{2}=A$, $A_{1}\cap A_{2}=\emptyset$
satisfying condition that either $B_{1}\cup X\subset A_{1}$ or $B_{2}\cup X\subset A_{2}$
exactly a half classifies $X$ as an object of sample $B_{1}$ and
a half classifies $X$ as an object from $B_{2}$. \end{theorem}
\begin{proof}
It is easy to define a one-to-one between classifications. Indeed,
if $\left\{ A_{1},A_{2}\right\} ,$ $A_{1}\cup A_{2}=A,$ is a classification
such that $B_{1}\subset A_{1},\, X\subset A_{1},\, B_{2}\subset A_{2},$
one maps it into the unique classification $\left\{ A_{1}^{\prime},A_{2}^{\prime}\right\} $
such that $B_{1}\subset A_{1}^{\prime},\, X\subset A_{2}^{\prime},\, B_{2}\subset A_{2}^{\prime},$
where $A_{1}^{\prime}=A_{1}\setminus X,\, A_{2}^{\prime}=A_{1}\cup X.$ \end{proof}
\begin{corollary}
A supervised pattern recognition is impossible. After the leaning
procedure the probability to classify correctly a new object is the
same as before leaning, i.e., $1/2$.
\end{corollary}

\section{Demonstration of algorithm}

A preliminary version of the forecast algorithm described above was
used in the paper {[}Ghertzik, 2008{]} for California and the Sumatra-Andaman
earthquake region. These computations serve as a demonstration of
the efficiency of the method but their actual results should be taken
with a pinch of salt because the selection of precursors does not
appear well-justified from the modern point of view: the number of
free parameters to be adapted in the precursor `\textquotedbl{}stress
indicator\textquotedbl{} is too large. 

\textbf{\textcolor{black}{Califormia region. }}\textcolor{black}{The
catalog Global Hypocenter Data Base CD-ROM NEIC/USGS, Denver, CO,
1989, together with data from the site NEIC/USGS PDE (ftp://hazard.cr.usgs.gov)
for earthquakes with magnitudes $M\ge4.0$ with epicenters between
$113^{\circ}-129^{\circ}$ of western longitude and $31^{\circ}-43^{\circ}$
of northern latitude was used for parameter adaptation. The initial
time $t_{0}$ was selected by subtracting from the time of actual
computation, $08.03.2006$, an integer number of half-year intervals
such that $t_{0}$ fits the first half of the year 1936. (The final
time $08.03.2006$ could be considered as an initial moment for constructing
half-year forecast forward up to the date of the latest earthquake
available in the catalog). During the computation the time interval
from the first half of 1936 to the first half of 1976 was used for
relaxation of the zero initial data used for precursors. After this
date the catalog for the earthquakes with magnitude $M\ge4.0$ was
used to estimate the probabilities of strong earthquakes with magnitude
$M\ge6.0$ up to the moment of actual forecast. Note that the adapted
restriction to include into consideration only earthquakes with magnitudes
$M\geq4.0$ is a severe restriction. It decreases the precision of
precursor computation and therefore, if a prediction is successful,
increases the degree of confidence to the predictor choice. We choose
$a=150$km as a size of the spacial lattice, and $\Delta t=6$ months
as a time-step. Retrospective forecast was performed with $5$ alarm
levels defined by the thresholds $a_{s}=10^{-\alpha(5-s)},s=1,\ldots,4$
and $\alpha=0.75$. (Due to too short time step no alarms was registered
on the lowest level when parameter $\alpha=1$ was selected). In order
to reduce the influence of the boundary conditions the large square
covering all the seismic events in the catalog used in the computations
was reduced by two layers of elementary cells from each boundary.
As a result of optimization the forecast information efficiency of
$0.526$ was achieved, i.e., the forecasting algorithm applied to
the given catalog extracted from it about $53\%$ of all available
information about seismic events. It seems that this result could
be only partially explained by a lucky selection of precursors: another
contributor to the high efficiency of the algorithm is the adaptation
of the parameters to the features of the specific catalog. The influence
of this artificial information may be reduced only with the increase
of the observation interval.}

\noindent \textcolor{black}{Accepting the rule of binary alarm announcement
in the cells from group 1 and 2 from 5 levels possible one obtains
that the space-time share of alarmed cells is $3.4\%$ and the share
of missed targets is $18.2\%$. This result is comparable with the
best forecasts available in the literature and obtained by other methods
(in the cases when the quantitive parameters of algorithms are presented
in the publications). When the forecast was constructed in the future
we obtained that the estimate of probability of a strong earthquake
anywhere in the area under study is $0.174$, and the maximal point
estimate of an event in any individual cell is $0.071$. As a whole
the seismic situation in California did not look too alarming. Indeed,
there were no strong earthquakes in California in the next half-year.}

\textbf{\textcolor{black}{SAE region}}\textcolor{black}{. We have
conducted a retrospective forecast of strong earthquakes with magnitudes
$M\ge7.0$ for the whole region where the Sumatra-Andoman earthquake
(SAE) happened on $26.12.2004$ with magnitude $M=9.3$. The catalog
Global Hypocenter Data Base CD-ROM NEIC/USGS, Denver, CO, 1989, together
with the data from the cite NEIC/USGS PDE (ftp://hazard.cr.usgs.gov)
for earthquakes with magnitudes $M\ge5.5$ with epicenters between
$84.3^{\circ}-128^{\circ}$ of eastern longitude and $20^{\circ}-26^{\circ}$
of northern latitude was used for the parameters adaptation. The initial
moment of time $t_{0}$ was selected by subtracting from the time
of actual computation, $10.11.2004$, an integer number of half-year
intervals such that $t_{0}$ fits the first half of the year 1936.
(The final time was selected in such a way that the next half-year
period covers SAE and its powerful aftershock). During the computation
the time interval from the first half of 1936 to the first half of
1976 was used for relaxation of the zero initial data used for precursors.
After this date the catalog was used to estimate the probabilities
of strong earthquakes with magnitude $M\ge7.0$ up to the moment of
actual forecast. (For magnitude $M\ge7.5$ the number of seismic events
was not sufficient for reliable forecast because the $5\%$-confidence
intervals strongly overlapped). In this case the restriction to include
into consideration only earthquakes with magnitudes $M\geq5.5$ was
adapted. As before, it decreases the precision of precursor computation
and therefore, if the prediction is successful, increases the degree
of confidence to the predictor choice. We selected the size of the
spacial grid as $a=400$km and the size of time-step $\Delta t=$half-year.
Retrospective analysis was conducted following the same scheme as
in the previous case. In order to reduce the influence of the boundary
conditions the large square covering all the seismic events in the
catalog used in the computations was reduced by two layers of elementary
cells from each boundary. As a result of optimization the forecast
information efficiency was $0.549$, i.e., the forecasting algorithm
extracted around $55\%$ of all available information about seismic
events when applied to the given catalog. }

\textcolor{black}{In the case of binary alarm announcement the space-time
share of alarmed cells was $3.1\%$ and the share of missed targets
was $8.3\%$. This result is comparable with the best forecasts available
in the literature and obtained by the other methods (in the cases
when the quantitive parameters of algorithms are presented in the
publications). In the case of forward forecast the two most powerful
earthquakes, i.e., SAE and its major aftershock, happened in the alarm
zone of the second level, and two other events with smaller magnitudes
in the fourth alarm zone. Note that in case of binary alarm announcement
one would register a square with 9 elementary cells with only one
alarmed and 8 quiet cells. In this case no reliable forward forecast
is possible.}
\begin{acknowledgements}
\textcolor{black}{\normalsize{We would like to thank V.Pisarenko and
G.Sobolev for stimulating discussions that gave us the impulse for
writing this paper.}}{\normalsize \par}\end{acknowledgements}

\end{document}